\documentclass[5p,twocolumn]{elsarticle}

\usepackage{mathrsfs,amsmath,amssymb,amsthm,graphicx,dsfont}
\usepackage{multirow,float,algorithm,algpseudocode}
\usepackage{mathtools,cuted,booktabs}
\usepackage{xcolor}

\theoremstyle{definition}

\newtheorem{assumption}{Assumption}
\newtheorem{lemma}{Lemma}
\newtheorem{theorem}{Theorem}
\newtheorem{remark}{Remark}
\newtheorem*{problem}{Problem}
\newtheorem{definition}{Definition}

\newtheorem{proposition}{Proposition}

\biboptions{sort&compress}

\makeatletter
\def\ps@pprintTitle{%
  \let\@oddhead\@empty
  \let\@evenhead\@empty
  \def\@oddfoot{\reset@font\hfil\thepage\hfil}
  \let\@evenfoot\@oddfoot
}
\makeatother

\begin{document}

\begin{frontmatter}

\title{Data-driven nonlinear output regulation via data-enforced incremental passivity}

\author{Yixuan Liu}
\ead{y.Liu-24@tudelft.nl}
\author{Meichen Guo\corref{cor1}}
\ead{m.guo@tudelft.nl}

\cortext[cor1]{Corresponding author}

\address{Delft Center for Systems and Control,
Delft University of Technology, 2628 CD Delft, The Netherlands}

\begin{abstract}

{This work proposes a data-driven nonlinear regulator design that achieves asymptotic reference tracking under external disturbances, where the reference and disturbances are generated by a linear exosystem.} The key idea is to design a data-driven feedback controller such that the closed-loop system is incrementally passive with respect to the regulation error and a virtual input. {By interconnecting the closed-loop system with an internal model and carefully designing the virtual input, we solve the data-driven nonlinear output regulation problem. We characterize the passivation feedback controller by a set of data-dependent linear matrix inequalities, which is independent of the internal model. This decoupled design offers high data efficiency and design flexibility. The proposed approach also solves the non-zero equilibrium stabilization problem of a class of nonlinear systems with unknown equilibrium input.} Numerical examples are presented to illustrate the effectiveness of the proposed designs.
\end{abstract}


\end{frontmatter}

\section{Introduction}

A primary challenge in controlling modern systems is the lack of mathematical models describing the dynamics of the systems. Many relevant systems found in fields such as neuroscience, epidemiology, and ecology have complex dynamics that are difficult to model using first principles or are too costly to identify a sufficiently accurate model. In applications where the major objective is to achieve control goals and an accurate model is of little or no interest, an efficient approach is the so-called direct data-driven control, where effective controllers are synthesized directly from data without explicitly identifying a model. Another challenge in controlling modern systems via data is providing rigorous performance guarantees in the presence of complex nonlinearities.

Direct data-driven control has been gaining increasing attention not only because it enables controller design directly from data with rigorous guarantees, but also because, in many settings, the resulting conditions can be formulated as computationally tractable semi-definite programs, as reviewed in \cite{Martin2023Survey}. Existing developments have primarily focused on the rigorous characterization of stabilizing controllers for both linear systems \cite{DePersis2020, Berberich_datadrivenMPC, vanWaarde2020DataInformativity} and nonlinear systems \cite{guo_Taylor, verhoek2023DDDLPV, koopman_NL,Guo2022,NL_cancellation}. In particular, for nonlinear systems, direct data-driven stabilization has been explored through, among other approaches, polynomial approximation \cite{guo_Taylor}, LPV embedding \cite{verhoek2023DDDLPV}, Koopman lifting \cite{koopman_NL}, and nonlinearity cancellation \cite{NL_cancellation}. In practice, applications often require more sophisticated control objectives beyond stabilization, including trajectory tracking. Motivated by this practical need, this work investigates a data-driven trajectory tracking problem for a class of nonlinear systems subject to external disturbances. 

Output regulation considers reference tracking and disturbance rejection of a plant, where the reference trajectory and disturbances are generated by an exosystem. Seminal results on linear and nonlinear output regulation can be found in \cite{FRANCIS1976457,Davison1976,Isidori1990,Huang2004}. The widely used output regulation framework augments the plant with an internal model and converts the output regulation problem into a stabilization problem of the augmented system. The authors of \cite{Pavlov2008IP&OPR} proposed an alternative approach using incremental passivity that describes the relationship between two arbitrary trajectories of a forced system. The approach involves designing a controller that makes the plant incrementally passive and interconnecting the plant with an incrementally passive internal model. The main feature of this approach is the decoupled design of the passivation controller and the internal model, which simplifies the overall regulator design. 

In addition to the model-based results mentioned above, output regulation of unknown systems has also been investigated by techniques such as adaptive dynamic programming \cite{gao2017ADPNonlinear}, neural networks \cite{lan2007NNOPR}, and reinforcement learning \cite{jiang2019RLOPR}. We note that most of these results either require a stabilizing control input to initiate the algorithm or demand a large amount of data without providing rigorous performance guarantees. Inspired by advancements in direct data-driven control, recent work has used finite offline data for output regulation. For example, \cite{TrentelmanDDOPR} established necessary and sufficient conditions for the informativity of a given set of data for the algebraic regulator problem. The authors of \cite{Zhu2024} later extended this result to linear robust output regulation via state feedback. {Assuming an unknown linear plant and disturbance model, \cite{Deutscher2024} used a Koopman operator to construct an extended system and achieved data-driven output regulation. Using data-based contractivity \cite{Hu2024Contraction}, recent work \cite{Hu2024DDOPR} proposed data-driven regulators for nonlinear systems. For the nonlinear case, the regulator in \cite{Hu2024DDOPR} guarantees bounded tracking errors for periodic disturbances and references. In summary, global asymptotic output regulation for nonlinear systems directly from data still requires further investigation.
}

This paper proposes a data-driven regulator design by enforcing incremental passivity via data. The proposed design contains two steps. The first step designs a data-driven state-feedback controller rendering the closed-loop system incrementally passive with respect to the regulation error and a virtual input. 
The second step interconnects the closed-loop system with an incrementally passive internal model designed based on the known exosystem and then obtains the overall regulator.


{\emph{\textbf{Contribution.}} We summarize the contributions of this work as follows. 
\begin{itemize}
    \item We solve the output regulation problem for a class of nonlinear plants directly from offline data in the presence of external disturbances and rigorously guarantee that the regulation error asymptotically converges to zero. This contrasts with recent results that either focus on linear plants or only achieve a bounded regulation error for nonlinear plants.
    \item The proposed data-driven output regulator consists of decoupled designs of the feedback passivation controller and the internal model. Specifically, the passivation controller is independent of the internal model, and the regulator design does not require data from the internal model, which makes the designs more efficient and flexible.
    \item  We characterize feedback incremental passivation controllers for a class of nonlinear systems
    . Based on this characterization, we derive data-based linear matrix inequalities (LMIs) for designing a feedback controller that renders the system incrementally passive. It is worth noting that in the direct data-driven setting, instead of enforcing passivity via control, most existing passivity-related results focus on data-based verification or analysis of dissipativity properties, such as in \cite{Koch2022DDDissVeri,Martin2023polyapprox,Verhoek2023DDDissipativityLPV}.
    \item Exploiting the proposed data-driven output regulation framework, we develop a data-driven solution to the nonlinear stabilization problem of a non-zero equilibrium, where the corresponding equilibrium input is neither known nor estimated.
\end{itemize}}

\emph{\textbf{Outline.}} The rest of the paper is organized as follows. Section \ref{sec:OPR&IP} formulates the data-driven nonlinear output regulation problem. Section \ref{sec:IP} presents model-based and data-driven designs of feedback controllers that make the closed-loop system incrementally passive. Section \ref{sec_DDORP} designs the internal model and the overall regulator by interconnecting the internal model and the closed-loop system obtained in the previous section. Based on the proposed output regulation approach, Section \ref{sec:stabilization} derives a result on data-driven nonlinear stabilization of a non-zero equilibrium with the unknown equilibrium input. Section \ref{sec_conclusion} summarizes the results of this work and discusses future research directions.

\emph{\textbf{Notation.}} Throughout the paper, $\mathbb R$ denotes the set of real numbers and $\mathbb R_{\ge 0}$ denotes the set of nonnegative real numbers. For a symmetric matrix $A$, $A\succ(\succeq)0$ indicates that $A$ is positive (semi-)definite, and $A\prec(\preceq)0$ indicates that $A$ is negative (semi-)definite. $A+(*)^{\top}$ represents the sum of a square matrix $A$ and its transpose. All elements of the vector $\mathds{1}_{n}\in\mathbb R^n$ are ones. $||x||$ denotes the Euclidean norm of a vector $x\in\mathbb R^n$. For matrices $A_1$, $A_2$, $\dots$, and $A_n$, $\mathrm{blockdiag}(A_1,A_2,\dots,A_n)$ denotes the matrix whose diagonal blocks are $A_1$ $A_2$, $\dots$, and $A_n$.


\section{Problem formulation}
\label{sec:OPR&IP}

Consider the nonlinear system
{
\begin{subequations}\label{system:NL}
\begin{align}
    \dot x &= AZ(x) +Bu +Ew \label{system:NLx}\\
    e & = CZ(x)  +Fw \label{system:NLe}
\end{align}
\end{subequations}
}
with state $x\in\mathbb R^n$, input $u\in\mathbb R^m$, regulation error $e\in\mathbb R^m$, and uncertainty $w\in\mathbb R^q$. {We assume that $Z:\mathbb R^n\rightarrow \mathbb R^{n_Z}$ is known and contains a library of continuous functions describing the unforced part of the nonlinear dynamics.} We write $Z$ in the form of $Z(x)=\begin{bmatrix} x^{\top} & Q(x)^{\top}\end{bmatrix}^{\top}$ where $Q(x)\in\mathbb R^{(n_Z-n)}$ denotes the nonlinear part of $Z$. The uncertain signal $w$ is generated by the exosystem
\begin{align}\label{exosystem}
    \dot w &= S w
\end{align}
where matrix $S\in\mathbb R^{q\times q}$ is known and the initial condition $w(t_0)\in\mathbb W\subset\mathbb R^q$ is unknown with $\mathbb W$ being a compact and forward-invariant set.

The output regulation problem considered in this work aims at designing a control input $u$ such that the regulation error $e$ converges to $0$ asymptotically despite the uncertainty $w$. {To ensure that the problem is solvable, we pose the following assumptions.}

\begin{assumption}\label{assump:S}
 {{The minimal polynomial of $S$ has only simple roots on the imaginary axis. }} 
\end{assumption}




\begin{assumption}\label{assump:RE}
    For any $w(t_0)\in\mathbb W$, {there exist $\mathcal X\in\mathbb R^{n\times q}$ and $\mathbf u(w)$ such that
    \begin{align}
    \label{RegulatorEquation}
        \mathcal X Sw &= AZ(\mathcal Xw) +B\mathbf u(w) +Ew\notag\\
        0 &= C Z(\mathcal Xw)+Fw.
    \end{align}}
\end{assumption}


\begin{remark}[On the assumptions] Assumption \ref{assump:S} is widely used in output regulation literature, which indicates that the exosystem in \eqref{exosystem} generates a combination of a step function with arbitrary amplitude and sinusoidal functions with arbitrary amplitudes and initial phases whose frequencies depend on the eigenvalues of $S$. 
{In \cite{Hu2024DDOPR}, the frequencies of the sinusoidal signals are required to be rationally related to ensure that the linear combination of exogenous signals is periodic. In contrast, Assumption \ref{assump:S} does not impose this condition, allowing the combined exogenous signal to include non-periodic components. 
Assumption \ref{assump:RE} serves as a feasibility condition ensuring that the asymptotic output regulation problem is solvable. In particular, it requires the state solution to be linear in $w$, which guarantees the existence of a linear internal model capable of solving the problem. Importantly, the proposed regulator in this work does not rely on the explicit steady-state solutions.
} 
\end{remark}

This work assumes that all the system matrices in \eqref{system:NL} are unknown, and a finite-length offline data set $\mathcal{DS}:=\{\dot x(t_k), x(t_k), u(t_k), e(t_k), k=0,1,\dots,T-1\}$ for some $T>1$ is sampled from one or multiple experiments. We formulate the data-driven output regulation problem as follows.
\begin{problem}[Data-driven output regulation]
    For nonlinear system \eqref{system:NL} with unknown system matrices and known exosystem \eqref{exosystem}, design a feedback control input $u$ using data set $\mathcal{DS}$, such that all the solutions to the closed-loop system are bounded and $\lim_{t\rightarrow \infty} e(t)=0$.
\end{problem}

\begin{remark}[Unknown nonlinear system]
    This work assumes that the vector $Z(x)$ consists of known nonlinearities, but the matrices $A$, $B$, $C$, $E$, and $F$ are unknown. {Existing results \cite{TrentelmanDDOPR} and \cite{Zhu2024} on data-driven output regulation assume that the matrices $C$, $E$, and $F$ are known.} For the approach proposed in this work, knowledge of $E$ and $F$ is not required for control design, and $C$ can be represented using the regulation error data.
\end{remark}

In what follows, we show that for exosystem \eqref{exosystem} with any initial condition, the sampled data of the exogenous signal $w$ can be arranged into a matrix that is the product of an unknown constant matrix and a known time-dependent matrix. This representation is in the same spirit as \cite[Section VI]{Hu2024Contraction} and is included here for the completeness of this paper.

Under Assumption \ref{assump:S}, we rearrange the exosystem such that $w$ and $S$ are partitioned as $w = \begin{bmatrix} w_1 & w_2 & \cdots & w_{2q_1}& w_{2q_1+1}&w_{2q_1+2}& \cdots & w_{2q_1 +q_2}\end{bmatrix}^{\top}$ and
\begin{align*}
    S=\mathrm{blockdiag}\left( \begin{bmatrix}
    0 & \sigma_1 \\ -\sigma_1 & 0 
\end{bmatrix}, \dots, \begin{bmatrix}
    0 & \sigma_{q_1} \\ -\sigma_{q_1} & 0 
\end{bmatrix}, 
0_{q_2\times q_2}\right)
\end{align*}
where $q_1$ and $q_2$ are nonnegative integers such that $2q_1+q_2=q$. We arrange the sampled data into the matrix 
\begin{align*}
W_0:=\begin{bmatrix} w(t_0) & w(t_1) & \cdots & w(t_{T-1}) \end{bmatrix}\in\mathbb R^{q\times T}, 
\end{align*}
which can be written as
\begin{align}\label{datamatrixW0}
    W_0 = \Gamma M_0
\end{align}
where the unknown matrix $\Gamma$ and known matrix $M_0$ are
\begin{align*}
    \Gamma &:=\mathrm{diag}\big(\gamma_1, \dots, \gamma_{q_1}, w_{2q_1+1}(t_0),\dots,w_{2q_1+q_2}(t_0) \big) \\
    &\quad\in\mathbb R^{q\times q} \\
    M_0 &:= \begin{bmatrix}
        \overline M(t_0) & \cdots & \overline M(t_{T-1})
    \end{bmatrix}\in\mathbb R^{q\times T} \\
    \overline M(t) &:= \begin{bmatrix}
        \overline M_1(t)^{\top} & \cdots & \overline M_{q_1}(t)^{\top} & \mathds{1}_{q_2}^{\top}
    \end{bmatrix}^{\top} \in\mathbb R^{q}\\
    \overline M_i(t) &:= \begin{bmatrix} \sin(\sigma_i t) & \cos(\sigma_i t)\end{bmatrix}^{\top} \in\mathbb R^2,\quad i=1,\dots,q_1
\end{align*}
for unknown constant matrices $\gamma_i\in\mathbb R^{2\times2}$, $i=1,\dots,q_1$.


\section{Incremental passivation via {state-}feedback}
\label{sec:IP}
This section first briefly reviews incremental passivity and then respectively presents model-based and data-driven conditions that characterize a static state-feedback controller that renders the nonlinear system \eqref{system:NL} incrementally passive with respect to the regulation error $e$ and a virtual input $v$. The results of this section address how to enforce incremental passivity via data-driven feedback and lay the foundation for solving the proposed data-driven output regulation problem.

\subsection{Incremental passivity}

Incremental passivity is an input-output property that holds for any two arbitrary trajectories generated by any two inputs. Following \cite{Pavlov2008IP&OPR}, we revisit the definition and interconnection properties of incrementally passive systems, which play an important role in deriving the main results.  

\begin{definition}[Incremental passivity \cite{Pavlov2008IP&OPR}]
\label{def:IP}
    The system
    \begin{align}\label{system:nonlinear}
        \dot x & = f(x,u,t) \notag\\
        y & = h(x,t)
    \end{align}
    with state $x\in\mathbb R^n$, input $u\in\mathbb R^m$, and output $y\in\mathbb R^m$ is incrementally passive if there exists a $C^1$ storage function $V(t,x_1, x_2):\mathbb R_{\ge 0}\times\mathbb R^{2n}\rightarrow \mathbb R_{\ge 0}$ such that for any two inputs $u_1(t)$ and $u_2(t)$ and any two solutions $x_1(t)$ and $x_2(t)$ of \eqref{system:nonlinear} corresponding to the inputs, the respective outputs $y_1=h(x_1,t)$ and $y_2=h(x_2,t)$ satisfy the inequality
    \begin{align}\label{equ:incre_pass_ineq}
        \dot V(t,x_1,x_2) &= \frac{\partial V}{\partial t}+\frac{\partial V}{\partial x_1}f(x_1,u_1,t) + \frac{\partial V}{\partial x_2} f(x_2,u_2,t) \notag\\
        &\le (y_1-y_2)^{\top}(u_1-u_2).
    \end{align}
\end{definition}

\begin{definition}[Regular storage function \cite{Pavlov2008IP&OPR}]
    A storage function $V(t,x_1,x_2)$ is called regular if for any sequence $(t_k,x_{1k},x_{2k})$, $k=1,2,\dots$, such that $x_{2k}$ is bounded, $t_k$ tends to infinity, and $\|x_{1k}\|\rightarrow +\infty$, it holds that $V(t_k,x_{1k},x_{2k})\rightarrow +\infty$ as $k\rightarrow +\infty$.
\end{definition} 

A simple example of a regular storage function is $V(x_1,x_2)=(x_1-x_2)^{\top}\mathcal P(x_1-x_2)$ with some $\mathcal P\succ 0$. 

Similarly to the conventional passivity property, the feedback interconnection of two incrementally passive systems with regular storage functions is also incrementally passive with a regular storage function. 

\begin{lemma} [\cite{Pavlov2008IP&OPR}]
\label{Lemma:interconnectionIP}
    Suppose the systems 
    \begin{subequations}
        \begin{align}
            \dot x &= F_x(x,u_x,t), \quad y_x = H_x(x,t) \label{sys:inter1} \\
            \dot z &=F_z(z,u_z,t), \quad y_z = H_z(z,t) \label{sys:inter2}
        \end{align}
    \end{subequations}
    are incrementally passive. Then the system that interconnects \eqref{sys:inter1} and \eqref{sys:inter2} through $u_x=\beta y_z +v_x$ and $u_z=-\beta^{\top}y_x+v_z$ with some square matrix gain $\beta$ is incrementally passive with respect to input $\overline v:=\begin{bmatrix} v_x^{\top} & v_z^{\top} \end{bmatrix}^{\top}$ and output $\overline y:=\begin{bmatrix} y_x^{\top} & y_z^{\top}\end{bmatrix}^{\top}$.
\end{lemma}




\subsection{Model-based incremental passivation}

Before presenting our main result on incremental passivation via data-driven state-feedback, we first address the model-based case to clearly reveal the underlying design, which leads to the data-driven result. It should be noted that designing a feedback controller to enforce conventional passivity has been extensively studied in the literature; for example, \cite{byrnes1991passivity,seron1994adaptivePassivation,jiang1996passification,larsen2001passivationOutput}. Making a system incrementally passive via output feedback has also been addressed in work such as \cite{Pavlov2008IP&OPR}. The objective of this subsection is to develop a model-based state-feedback passivation approach that benefits the subsequent extension to data-driven passivation. 

Consider the static state-feedback controller
\begin{align}\label{controller:Zx}
    u & =  KZ(x) + v
\end{align}
where $v\in\mathbb R^m$ is a virtual input and $K$ is the control gain to be designed. Applying controller \eqref{controller:Zx} to  plant \eqref{system:NL} results in the closed-loop system 
\begin{align}\label{system:CLnonlinear}
 \dot x &= (A+BK) Z(x) + B v +Ew \notag\\
    e & = CZ(x) +Fw.
\end{align}

Using Definition \ref{def:IP}, a characterization of incremental passivity for system \eqref{system:CLnonlinear} is derived.

\begin{lemma}\label{lemma:IPZx}
    If there exist $K\in\mathbb R^{m\times n_Z}$ and a positive definite matrix $\mathcal P\in\mathbb R^{n\times n}$ such that
    \begin{subequations} 
    \label{conditions:IPZx_nox}
    \begin{align}
          \mathcal I^{\top}\mathcal P(A+BK) + (A+BK)^{\top}\mathcal P\mathcal I &\preceq 0 \label{condition:IPZxA}\\
         \mathcal I^{\top} \mathcal P B &=  C^{\top}, \label{condition:IPZxBC}
    \end{align}
    \end{subequations}
    where $\mathcal I := \begin{bmatrix}
          I_n & 0_{n\times(n_Z-n)} 
        \end{bmatrix}$, 
    then the closed-loop system \eqref{system:CLnonlinear} is incrementally passive with respect to input $v$, output $e$, and a regular storage function. 
\end{lemma}
\begin{proof}
    Define the regular storage function as $V(x_1,x_2)=\frac{1}{2}(x_1-x_2)^{\top}\mathcal P (x_1-x_2)$. The time derivative of the storage function along the solutions $x_1$ and $x_2$ corresponding to the inputs $v_1$ and $v_2$ satisfies that
    \begin{align*}
        &\quad \dot V(x_1,x_2)\\
        & =\frac{1}{2} (x_1-x_2)^{\top}\mathcal P(\dot x_1-\dot x_2) + \frac{1}{2}(\dot x_1-\dot x_2)^{\top}\mathcal P(x_1-x_2) 
        \\
        & = \frac{1}{2}(x_1-x_2)^{\top} \mathcal P \Big( (A+BK) \big(Z(x_1)-Z(x_2)\big)\\
        &\quad +B (v_1-v_2)\Big ) + (*)^{\top}.
    \end{align*}
Recalling $Z(x)=\begin{bmatrix} x^{\top} & Q(x)^{\top}\end{bmatrix}^{\top}$ and by the definition of $\mathcal I$, one has that
\begin{align*}
    &\quad x_1- x_2 
     =  \mathcal I \begin{bmatrix}
         x_1-x_2 \\ Q(x_1)-Q(x_2)
    \end{bmatrix}
     =  \mathcal I \big( Z(x_1)-Z(x_2)\big).
\end{align*}
Therefore, it holds that 
\begin{align*}
     \dot V(x_1,x_2)
    & =\frac{1}{2}\big(Z(x_1)- Z(x_2)\big)^{\top} \Big (\mathcal I^{\top}\mathcal P(A+BK) \\
    &\quad + (A+BK)^{\top}\mathcal P \mathcal I \Big) \big(Z(x_1)-Z(x_2)\big) \\
    &\quad +\big( Z(x_1)- Z(x_2)\big)^{\top} \mathcal I^{\top}\mathcal P B (v_1-v_2).
\end{align*}
The difference between the corresponding $e_1$ and $e_2$ is
\begin{align*}
    &\quad~~ e_1-e_2 \\
    &= CZ(x_1) +Fw- C Z(x_2) - Fw = C \big(Z(x_1)- Z(x_2)\big).
\end{align*}
Under conditions \eqref{condition:IPZxA} and \eqref{condition:IPZxBC}, one has that
\begin{align*}
    \dot V(x_1,x_2) \le (e_1-e_2)^{\top}(v_1-v_2).
\end{align*}
This shows that system \eqref{system:CLnonlinear} is incrementally passive with respect to input $v$, output $e$, and regular storage function $V(x_1,x_2)$.
\end{proof}

{ 


Lemma \ref{lemma:IPZx} provides a feedback incremental passivation control characterization for the nonlinear system in \eqref{system:NL}. We note that the incremental passivity conditions presented in \cite[Lemma 3]{Pavlov2008IP&OPR} is not applicable to \eqref{system:NL}, as the system output therein depends linearly on $x$.

}

\begin{remark}[Passivation of uncertain systems]
\label{remark:RobustPassivation}
    In the case where the system matrices in \eqref{system:NL} contain uncertain parameters, the conditions characterizing incremental passivity need to compensate for the uncertainties. If the bounds of the uncertainties are available, one can derive a robust version of \eqref{condition:IPZxA} to hold for all uncertainties within the known bound, but the same is difficult for \eqref{condition:IPZxBC}. We note that passivation of uncertain systems have been addressed in work such as \cite{seron1994adaptivePassivation,jiang1996passification,peaucelle2008passificationRobust} using adaptive control techniques and \cite{xie1998passivityUncertain} via integral quadratic constraints (IQCs). {Applying techniques such as passivity indices may also reduce the fragility caused by \eqref{condition:IPZxBC} for robust passivation of uncertain systems.} 
\end{remark}

\subsection{Data-driven incremental passivation }

In this subsection, we extend the model-based conditions to data-based ones that characterize the incremental passivity of closed-loop system \eqref{system:CLnonlinear} using the data set $\mathcal{DS}$. To achieve this, we first use the data set to represent the closed-loop dynamics in \eqref{system:CLnonlinear}. 

We obtain the following data matrices
  \begin{align*}
    X_0 & := \begin{bmatrix}
               x(t_0) & x(t_1) & \cdots & x(t_{T-1})
             \end{bmatrix}\in\mathbb R^{n\times T} \\
    Z_0 &:= \begin{bmatrix}
           Z(x(t_0)) & Z(x(t_1)) & \cdots & Z(x(t_{T-1}))
         \end{bmatrix}\in\mathbb R^{n_Z\times T} \\
    X_1 & := \begin{bmatrix}
              \dot x(t_0) & \dot x(t_1) & \cdots & \dot x(t_{T-1})
             \end{bmatrix}\in\mathbb R^{n\times T} \\
    E_0 & := \begin{bmatrix}
               e(t_0) & e(t_1) & \cdots & e(t_{T-1})
             \end{bmatrix}\in\mathbb R^{m\times T} \\
    U_0 & := \begin{bmatrix}
               u(t_0) & u(t_1) & \cdots & u(t_{T-1})
             \end{bmatrix}\in\mathbb R^{m\times T}
  \end{align*}
from $\mathcal{DS}$ and the known vector $Z(x)$. The following lemma derives a representation of \eqref{system:CLnonlinear} using the data matrices.

\begin{lemma}\label{lemma:data_driven_repre}
    Consider nonlinear system \eqref{system:NL}, controller \eqref{controller:Zx}, and data set $\mathcal{DS}$. For any matrices $K\in\mathbb R^{m\times n_Z}$, and $G=\begin{bmatrix} G_1 & G_2\end{bmatrix}\in\mathbb R^{T\times(n_Z+m)}$ with $G_1\in\mathbb R^{T\times n_Z}$ and $G_2\in\mathbb R^{T\times m}$ satisfying 
    \begin{align}\label{condition:CLdatarep}
        \begin{bmatrix}
        I_{n_Z}  & 0_{n_Z\times m} \\
        K  & I_m \\
        0_{q\times n_Z} & 0_{q\times m} 
    \end{bmatrix}
    =\begin{bmatrix}
        Z_0 \\ U_0 \\ M_0
    \end{bmatrix}G,
    \end{align}
    the closed-loop system composed by \eqref{system:NL} and \eqref{controller:Zx} is
    \begin{align}\label{system:CLData}
    \dot x &= A_{\mathrm{d}} Z(x) + B_{\mathrm{d}} v + Ew \notag\\
    e & = C_{\mathrm{d}} Z(x) +Fw
    \end{align}
    where the data-based system matrices are defined as 
    \begin{align} \label{datarepresentation:Zx}
      A_{\mathrm{d}} =X_1G_1,~B_{\mathrm{d}} = X_1G_2,~ C_{\mathrm{d}} = E_0G_1.
\end{align}
\end{lemma}
\begin{proof}
    By the dynamics of system \eqref{system:NL}, the data matrices satisfy that
    \begin{align*}
        X_1 & = AZ_0 + BU_0 +EW_0\\
        E_0 & = CZ_0 +FW_0
    \end{align*}
    where $W_0=\Gamma M_0$ based on the analysis in Section \ref{sec:OPR&IP}. Under \eqref{condition:CLdatarep}, it holds that
    \begin{align*}
             A&+BK\\
            &\! = \!
        \big[ A ~ B ~ E\Gamma \big]\begin{bmatrix}
            I_{n_Z}\\  K \\ 0
        \end{bmatrix} = 
         \big[ A ~ B ~ E\Gamma \big]\begin{bmatrix}
        Z_0 \\ U_0 \\ M_0
    \end{bmatrix} G_1  =  X_1G_1, 
    \end{align*} 
    {and similarly, $B=X_1G_2$ and $C=E_0G_1$.} Recall the dynamics of \eqref{system:CLnonlinear}, one can obtain the data-based system matrices given in \eqref{datarepresentation:Zx}.
\end{proof}

Having represented \eqref{system:CLData} using data set $\mathcal DS$, we give the following result for the data-driven design of controller \eqref{controller:Zx}. 

\begin{theorem}\label{theorem:DDIPZx}
     Consider system \eqref{system:NL}, state feedback controller \eqref{controller:Zx}, and data set $\mathcal{DS}$. If there exist matrices $Y\in\mathbb R^{T\times n_Z}$, $G_2\in\mathbb R^{T\times m}$, and positive definite matrix $ P=\mathrm{blockdiag}(P_1,P_2)\in\mathbb R^{n_Z\times n_Z}$ with $P_1\in\mathbb R^{n\times n}$ and $P_2\in\mathbb R^{(n_Z-n)\times (n_Z-n)}$, such that
     \begin{subequations}\label{condition:DDIPZx}
        \begin{align}
        \begin{bmatrix}
            P \\ 0_{q\times n_Z}  \end{bmatrix} &=\begin{bmatrix}
        Z_0 \\ M_0
    \end{bmatrix}Y & \label{condition:ZxY1Y2} \\
       \begin{bmatrix}
            0_{n_Z\times m} \\ I_m \\ 0_{q\times m}
        \end{bmatrix} &= \begin{bmatrix}
        Z_0 \\ U_0 \\ M_0
    \end{bmatrix}G_2  \label{condition:ZxG3} \\
   \mathcal{I}^{\top}X_1Y +(*)^{\top}& \preceq 0  \label{condition:ZxAP} \\
    \begin{bmatrix}
        (X_1G_2)^{\top} & 0_{m\times (n_Z-n)}
    \end{bmatrix}&=  E_0Y, \label{condition:ZxBC}
        \end{align}
    \end{subequations}
    then controller \eqref{controller:Zx} with 
    \begin{align} \label{control: data_driven_expression}
        K=U_0YP^{-1}
    \end{align}
    renders the closed-loop system incrementally passive with respect to input $v$ and output $e$ with a regular storage function. 
\end{theorem}

\begin{proof}
   Let $G_1= YP^{-1}$, then conditions \eqref{condition:ZxY1Y2} and \eqref{condition:ZxG3} guarantee the existence of the data-based closed-loop representation \eqref{datarepresentation:Zx}. 

   Left- and right-multiply $P^{-1}$ to both sides of \eqref{condition:ZxAP} gives
 \begin{align*}
    &\quad~ P^{-1}\mathcal{I}^{\top}
        X_1Y P^{-1}+(*)^{\top}\\
    & = P^{-1}\mathcal{I}^{\top}
        X_1G_1 P  P^{-1}+(*)^{\top}\\
    &= \begin{bmatrix}
        P_1^{-1} & 0 \\ 0 & P_2^{-1}
    \end{bmatrix}\mathcal I^{\top}
        X_1G_1 +(*)^{\top}\\
    &= \begin{bmatrix}
        P_1^{-1} \\ 0
    \end{bmatrix} 
        X_1G_1 +(*)^{\top} \preceq 0 
    \end{align*} 
   Define the regular storage function as $V(x_1,x_2)=\frac{1}{2}(x_1-x_2)^{\top}P_1^{-1} (x_1-x_2)$. By the definition of $\mathcal I$ and the data-based representation $A_{\mathrm{d}}$, one has that
   \begin{align*}
       \mathcal I^{\top}P_1^{-1} A_{\mathrm{d}} = \begin{bmatrix}
        P_1^{-1} \\ 0
    \end{bmatrix} 
        X_1G_1 .
   \end{align*}
   Therefore, condition \eqref{condition:ZxAP} ensures that \eqref{condition:IPZxA} in Lemma \ref{lemma:IPZx} holds.

  At the both sides of condition \eqref{condition:ZxBC}, we right-multiply $P^{-1}$ and obtain
   \begin{align*}
        \begin{bmatrix}
        (X_1G_2)^{\top}P_1^{-1} & 0_{m\times (n_Z-n)}
    \end{bmatrix}&=
        E_0G_1 .
   \end{align*}
   Again, by the definition of $\mathcal I$ and the data-based representations $B_{\mathrm{d}}$ and $C_{\mathrm{d}}$, one has that
   \begin{align*}
       \mathcal I^{\top}P_1^{-1} B_{\mathrm{d}} &=\begin{bmatrix}
           (X_1G_2)^{\top}P_1^{-1} &  0_{m\times (n_Z-n)}
       \end{bmatrix}^{\top} \\
       C_{\mathrm{d}}  &=E_0G_1.
   \end{align*}
   Therefore, condition \eqref{condition:ZxBC} ensures that \eqref{condition:IPZxBC} in Lemma \ref{lemma:IPZx} holds, and thus the proof is complete.
\end{proof} 

{ 
\begin{remark} [Data requirement and computational complexity] \label{remark:DataRequirement}
To ensure that a data-based closed-loop system representation in Lemma \ref{lemma:data_driven_repre} exists, a sufficient condition is that the data matrix $\begin{bmatrix} Z_0^\top & U_0^\top & M_0^\top \end{bmatrix}^\top \in \mathbb{R}^{(n_Z +m+q) \times T}$ has full row rank. This rank condition holds only if $T\ge n_Z+m+q$. After data collection, this rank condition can be directly verified. If the condition is not satisfied, one may either collect additional data or redesign the input signal used for the experiment.

The computational complexity of \eqref{condition:DDIPZx} depends on the number of basis functions $n_Z$, input dimension $m$, and data length $T$. Increasing $T$ may help satisfy the rank condition, but it also increases the size of the decision variables and the LMIs in \eqref{condition:DDIPZx}. Therefore, one should choose $T$ to be sufficiently large to satisfy the rank condition, while avoiding unnecessarily long data sequences that increase computational cost.
\end{remark}
\vspace{-1ex}

\begin{remark}[Nonlinearity elimination]
Theorem \ref{theorem:DDIPZx} handles the nonlinear terms $Q(x)$ by assuming known nonlinearities and eliminating their effect, which enables global incremental passivation. We note that the feasibility of perfect nonlinearity elimination depends on whether these nonlinearities are in matched channels with the control input. Relaxing this matching-condition requirement via backstepping techniques and further relaxing the assumptions of the known nonlinearities are important directions for our future research.
\end{remark}
}

In this work, we assume that the experimental data are not subject to measurement noise. In fact, if the measurement noise can be generated by an exosystem in the form of \eqref{exosystem}, such as constant or periodic noise with known frequencies, then the proposed design is still applicable by updating the exosystem to include the noise. { When a more general class of noise is considered or the state derivative is estimated via numerical differentiation, data-based dynamics \eqref{system:CLData} will contain uncertainties, and a robust passivation approach is needed.} The difficulties of robust passivation under the current framework and possible solutions have been discussed in Remark \ref{remark:RobustPassivation}.


\section{Data-driven output regulator design}
\label{sec_DDORP}

The previous section addresses how to design a data-driven feedback controller that results in an incrementally passive closed-loop system. As incremental passivity describes the relationship between any two arbitrary trajectories of a forced system, it is an effective tool for solving output regulation problems. This section designs an incrementally passive internal model and interconnects it with the closed-loop system to obtain an incrementally passive augmented system. Then, we present the design of the regulator that solves the proposed output regulation problem.

We design the incrementally passive internal model as
\begin{align}\label{internalmodel}
    \dot \eta &= S \eta + \alpha \Xi \tilde e \notag \\
    \tilde v &= \Xi^{\top}\eta
\end{align}
where $\eta\in\mathbb R^q$, $\tilde e=-e$, and any fixed $\alpha>0$. We select $\Xi\in\mathbb R^{q\times m}$ as any nonzero matrix such that the pair $(S, \Xi)$ is controllable. The corresponding regular storage function is $V_{\mathrm{IM}}(\eta_1,\eta_2):=\frac{1}{2\alpha}\|\eta_1-\eta_2\|^2$. Figure \ref{fig:diagramCL} illustrates the feedback interconnection of the closed-loop system and internal model \eqref{internalmodel}. 

\begin{figure}[h]
\centering 
\includegraphics[trim=3.7cm 17.8cm 3.7cm 3.5cm, clip, scale=0.65]{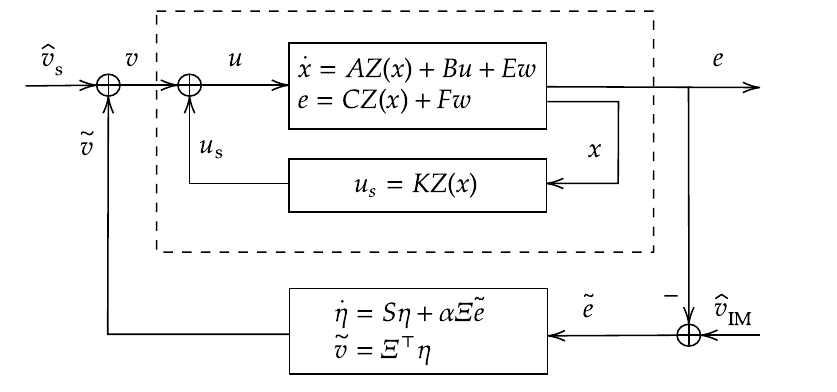}
\caption{Interconnection of the closed-loop system and the internal model.}
\label{fig:diagramCL}
\end{figure}

\begin{theorem}\label{theorem:nl}
    Under Assumptions \ref{assump:S} and \ref{assump:RE}
    , consider system \eqref{system:NL}, exosystem \eqref{exosystem}, and data set $\mathcal{DS}$. If there exist matrices $Y\in\mathbb R^{T\times n_Z}$, $G_2\in\mathbb R^{T\times m}$, and positive definite matrix $P=\mathrm{diag}(P_1,P_2)\in\mathbb R^{n_Z\times n_Z}$ with $P_1 \in\mathbb R^{n\times n}$ and $P_2 \in\mathbb R^{(n_Z-n)\times (n_Z-n)}$, such that \eqref{condition:DDIPZx} holds, then the regulator
    \begin{align}\label{controller:overall}
        \dot \eta &= S\eta -\alpha\Xi e \notag\\
        u &= KZ(x)+\Xi^{\top}\eta -\widehat K e. 
    \end{align}
    where $K=U_0Y P^{-1}$, any $\alpha>0$, and any positive definite $\widehat K\in\mathbb R^{m\times m}$, solves the output regulation problem.
\end{theorem}

\begin{proof}
    Firstly, as illustrated in Figure \ref{fig:diagramCL}, the closed-loop system and the internal model are interconnected via
    \begin{equation} \label{interconnection}
        v = \tilde v + \hat{v}_\text{s}, \ \widetilde e = -e+\hat{v}_\text{IM}.
    \end{equation}
    The resulting interconnected system has the following dynamics 
    \begin{subequations}\label{system:augument}
        \begin{align}
    \dot \zeta &= \begin{bmatrix}
        A+BK & B \Xi^{\top}\\
        -\alpha\Xi C & S
    \end{bmatrix} \widetilde Z(\zeta)+\begin{bmatrix}
    B & 0_{n\times m} \\0_{q\times m} & \alpha \Xi
    \end{bmatrix} \hat{v} \\
    &\quad~+\begin{bmatrix}
        E \\ -\alpha \Xi F
    \end{bmatrix}w \notag\\
    \hat{e} & = \begin{bmatrix}
        C & 0_{n \times q} \\ 0_{m\times n_z} & \Xi^{\top}
    \end{bmatrix}\widetilde Z(\zeta) +\begin{bmatrix}
        F \\0_{m \times q} 
    \end{bmatrix}w
 \end{align} 
    \end{subequations}
where $\zeta:=[x^\top \ \eta^\top]^\top$, $ \widetilde Z(\zeta) =  \begin{bmatrix}
        Z(x)^\top & \eta ^\top
    \end{bmatrix}^\top$, input $\hat v:=\begin{bmatrix} \hat v_\text{s}^{\top} & \hat v_\text{IM}^{\top} \end{bmatrix}^{\top}$, and output $\hat e :=\begin{bmatrix} e^{\top} & \tilde v^{\top}\end{bmatrix}^{\top}$.
    
    By Theorem \ref{theorem:DDIPZx}, the designed feedback controller \eqref{controller:Zx} satisfying \eqref{condition:DDIPZx} renders the closed-loop system incrementally passive with input $v$, output $e$, and regular storage function $V(x_1,x_2)$. Additionally, internal model \eqref{internalmodel} is incrementally passive with input $\tilde e$, output $\tilde v$, and regular storage function $V_\text{IM}(\eta_1,\eta_2)$. Therefore, by Lemma \ref{Lemma:interconnectionIP}, augmented system \eqref{system:augument} is incrementally passive with input $\hat v$, output $\hat e$, and regular storage function $V_\text{aug}$ defined as
     \begin{align}\label{equ_stroge_aug}
        V_{\text{aug}}(\zeta_1, \zeta_2):=(\zeta_1-\zeta_2)^\top \begin{bmatrix}
        \frac{1}{2}P_1^{-1} & 0_{n \times q} \\0_{q \times n} & \frac{1}{2 \alpha}I_{q}
    \end{bmatrix}(\zeta_1-\zeta_2).
    \end{align}
    Close the loop of \eqref{system:augument} by $\hat{v}=-\overline{K}\hat{e}$ where
    \begin{align}\label{controller:finalfb}
        \hat v = \begin{bmatrix}
            \hat v_\text{s} \\ \hat v_\text{IM}
        \end{bmatrix} = -\underbrace{\begin{bmatrix}
            \widehat K & 0_{m\times m} \\ 0_{m\times m} & 0_{m\times m}
        \end{bmatrix}}_{{:=\overline{K}}}  \begin{bmatrix}
            e \\ \tilde v
        \end{bmatrix} = \begin{bmatrix}
            -\widehat K e \\0_{m\times 1}
        \end{bmatrix}
    \end{align}
    with any $\widehat K\succ 0$, which leads to the overall regulator \eqref{controller:overall}. 
    
    Next, we prove that \eqref{controller:overall} solves the output regulation problem. {The designed passivation controller eliminates the nonlinear effects of the plant \eqref{system:NL}, and hence results in linear closed-loop state dynamics. Under Assumption \ref{assump:RE}
    , the closed-loop system \eqref{system:CLnonlinear} with input $v$ and output $e$ has the steady-state solution $\mathbf x(w)=\mathcal Xw$ and $\mathbf v(w)=\mathcal Vw$ for some $\mathcal V\in\mathbb R^{m\times q}$ corresponding to $e=0$. Therefore, system \eqref{system:augument} admits a bounded steady-state solution 
    $\zeta_{\mathrm{ss}}(t) = \big( \mathbf x(w), \eta_{\mathrm{ss}}(t)\big)$ corresponding to steady-state error $e=0$.}
    Recalling that \eqref{system:augument} is incrementally passive and the incremental passivity inequality \eqref{equ:incre_pass_ineq}, substitute  $\big(\zeta(t), \hat{v}(t), \hat{e}(t) \big)$ for $(x_1, y_1, u_1)$ and substitute steady-state solution $\big(\zeta_{\mathrm{ss}}(t), \hat v_{\mathrm{ss}}(t), \hat{e}_{\mathrm{ss}}(t) \big)$ where ${\hat e_\text{ss}}(t) := \begin{bmatrix}0_{1 \times m}&  \eta^\top_{\mathrm{ss}}(t)\Xi  \end{bmatrix}^\top$ for $(x_2, u_2, y_2) $. 
    Then, by the definitions of $\overline K$, $\hat e(t)$, and $\hat e_\text{ss}(t)$, 
    we derive that
    \begin{align*}
      \dot{V}_{\text{aug}}\big(\zeta(t), {\zeta}_{\mathrm{ss}}(t)\big) & \leq -\big(\hat{e}(t)-{\hat e_\text{ss}}(t) \big)^\top\overline K \big(\hat{e}(t)-{\hat e_\text{ss}}(t) \big) \notag \\
     & = -e(t)^\top \widehat K e(t) \le 0.
    \end{align*} 
    Then, we can further obtain that
    \begin{align*}
       &\quad~~ V_{\text{aug}}\big(\zeta(t), {\zeta}_{\mathrm{ss}}(t)\big) -V_{\text{aug}}\big(\zeta(t_0), {\zeta}_{\mathrm{ss}}(t_0)\big)\\
       &\leq -\int_{t_0}^t {e}(\tau)^\top\widehat K{e}(\tau) d\tau 
        \leq 0.
    \end{align*}
Thus, it holds that $V_{\text{aug}}\big(\zeta(t), {\zeta}_{\mathrm{ss}}(t)\big) \leq V_{\text{aug}}\big(\zeta(t_0), {\zeta}_{\mathrm{ss}}(t_0)\big)$. Boundedness of $V_{\text{aug}}\big(\zeta(t),{\zeta}_{\mathrm{ss}}(t)\big)$ and ${\zeta}_{\mathrm{ss}}(t)$ implies that $\zeta(t)$ is bounded for all $t\ge0$. Therefore, by the definition of $\zeta(t)$, we conclude that $x(t)$ is bounded on $\mathbb{R}_{\geq 0}$. Besides, under Assumption \ref{assump:S}, $w(t)$ and $\dot w(t)$ are bounded on $\mathbb{R}_{\geq 0}$. Thus, recalling the dynamics of plant \eqref{system:NL}, $\dot x(t) $ is bounded on $\mathbb{R}_{\geq 0}$. It then can be deduced that the derivative of ${e}(t)$ is bounded, which implies the uniform continuity of ${e}(t)^\top \widehat K{e}(t)$ on $\mathbb{R}_{\geq 0}$. As we have shown that $\dot V_{\text{aug}}\big(\zeta(t), {\zeta}_{\mathrm{ss}}(t)\big)<0$ when $\zeta(t) \neq {\zeta}_{\mathrm{ss}}(t)$, it holds that 
\begin{equation*}
    \int_{t_0}^{+\infty} {e}(\tau)^\top \widehat K{e}(\tau) d\tau \leq V_{\text{aug}}\big(\zeta(t_0), \zeta_{\mathrm{ss}}(t_0)\big) < +\infty. 
\end{equation*}
Via Barbalat's lemma \cite[Lemma 8.2]{book_Khalil}, we conclude that $\lim_{t\rightarrow +\infty} {e}^\top \widehat K {e}(t) = 0$. As $\widehat K\succ 0$, it holds that $\lim_{t\rightarrow +\infty}e(t)=0$. Thus, the output regulation problem is solved by regulator \eqref{controller:overall}.
\end{proof}

{One feature of the proposed data-driven output regulation approach is that the designs of controller \eqref{controller:Zx} and internal model \eqref{internalmodel} are decoupled, which offers high data efficiency and design flexibility.} This design is different from the traditional output regulation framework where the controller is designed for the entire augmented system. In the data-driven setting, obtaining and controlling the augmented system can further complicate the problem. For example, the signal generated by the internal model also needs to be sampled to derive the data-based representation of the augmented system in \cite{Hu2024DDOPR}. Moreover, if the disturbance affecting the plant or the reference changes, i.e., the exosystem changes, one must redesign the entire regulator based on the redesigned internal model. In contrast, the proposed approach does not require a complete redesign of the regulator, as conditions \eqref{condition:DDIPZx} characterizing the incremental passivation controller remain the same, and only the data set $\mathcal{DS}$ needs to be recollected to capture the effect of the new exosystem on the plant. 
\vspace{-1ex}



{
\begin{remark}[Linear internal model]
    The existence of $\mathcal X$ in Assumption \ref{assump:RE} guarantees the sufficiency of the linear internal model in \eqref{internalmodel} for solving the output regulation problem. In some cases, such as the stabilization problem addressed in Section \ref{sec:stabilization} or when the regulation error depends linearly on the state, this assumption readily holds. Nonetheless, it is of our interest to relax the assumption by designing a nonlinear internal model.
\end{remark}
\vspace{-0.7em}

\begin{remark}[Comparison with existing data-driven methods]
    For systems having unknown dynamics, various approaches have been applied to output regulation, such as adaptive dynamic programming \cite{gao2017ADPNonlinear}, neural networks \cite{lan2007NNOPR}, and reinforcement learning \cite{jiang2019RLOPR}. While adaptive dynamic programming often requires a stabilizing controller to initiate the algorithm, and neural networks require a large amount of data and may lack theoretical guarantees, our proposed approach seeks an alternative that removes these restrictions at the price of offline open-loop data and knowledge on the nonlinearities. We aim to relax the assumptions of the system dynamics and the experimental data in our future research. 
    
Theorem \ref{theorem:nl} differs from the contraction-based approach in \cite{Hu2024DDOPR} in disturbance requirement, regulation performance, data requirement, and nonlinearity assumptions. 
By enforcing incremental passivity on the closed-loop system, the proposed approach achieves global asymptotic tracking despite time-varying disturbances that are not necessarily periodic. Thanks to the decoupled designs, we do not need to collect data from the internal model. On the other hand, the contraction-based approach in \cite{Hu2024DDOPR} assumes periodic references and disturbances. The regulator design therein requires data from the internal model. The approach guarantees a bounded regulation error by suppressing finitely many harmonics of the periodic disturbance. We also note that this work handles nonlinearity differently from \cite{Hu2024DDOPR}. While this work assumes that perfect nonlinearity cancellation is feasible, \cite{Hu2024DDOPR} assumes known global Jacobian bounds on the nonlinearities.
\end{remark}
}

In what follows, we provide a numerical example of the proposed data-driven output regulation approach.

\textbf{Example 1.} Consider a simple pendulum that has the dynamics
\begin{align*}
    \dot x_1 &= x_2+d_1\\
    \dot x_2 & = -10\sin(x_1)-x_2+10u +d_2 \\
    e & = x_2-r,
\end{align*}
where $x_1$ and $x_2$ are the angle and the angular velocity of the pendulum, and $u$ is the applied torque. We denote the disturbance as $d=[d_1 ~~ d_2]^{\top}$ and the desired reference as $r$, both of which are linear functions of the exogenous signal $w$. In this example, $w$ is governed by the dynamics
\begin{align} \label{sys: e.g. exosys}
    \dot w &=\begin{bmatrix}
        0 & 2 &0 \\ -2 & 0 & 0 \\ 0 & 0 & 0
    \end{bmatrix} w,
\end{align}
with the initial condition $w(t_0)=[0~~1~~1]^{\top}$. The disturbances and reference are $d_1=\cos(2t+\frac{\pi}{3})$, $d_2=1$ and $r=\sin (2t)$.
It is noted that the explicit dynamics are only used for data collection.

Let $Z(x)=\begin{bmatrix} x_1 & x_2 & \sin(x_1)\end{bmatrix}^{\top}$ and conduct an experiment with $x(t_0)=[-0.1~~0.1]^{\top}$, and $u=\sin(t)$. {The data set is collected with the sampling period of $0.5$s, and the length of the data matrices is $T=20$. Solving condition \eqref{condition:DDIPZx} using YALMIP \cite{Loefberg2004} with the MOSEK solver \cite{Mosek2023} gives $K=[-0.1647 ~~ 0.0269 ~~ 1.0000]$. Setting $\alpha=5$, $\Xi=[1~~0~~1]^{\top}$, and $\widehat K=0.5$, we obtain the overall regulator.} Simulation results of regulation error $e$ from different initial conditions are illustrated in Figure \ref{fig:pendulum}, which show that the errors converge to $0$ in all cases.

\begin{figure}[h]
\centering 
\includegraphics[trim=1.8cm 7.1cm 2cm 7.5cm, clip, scale=0.45]{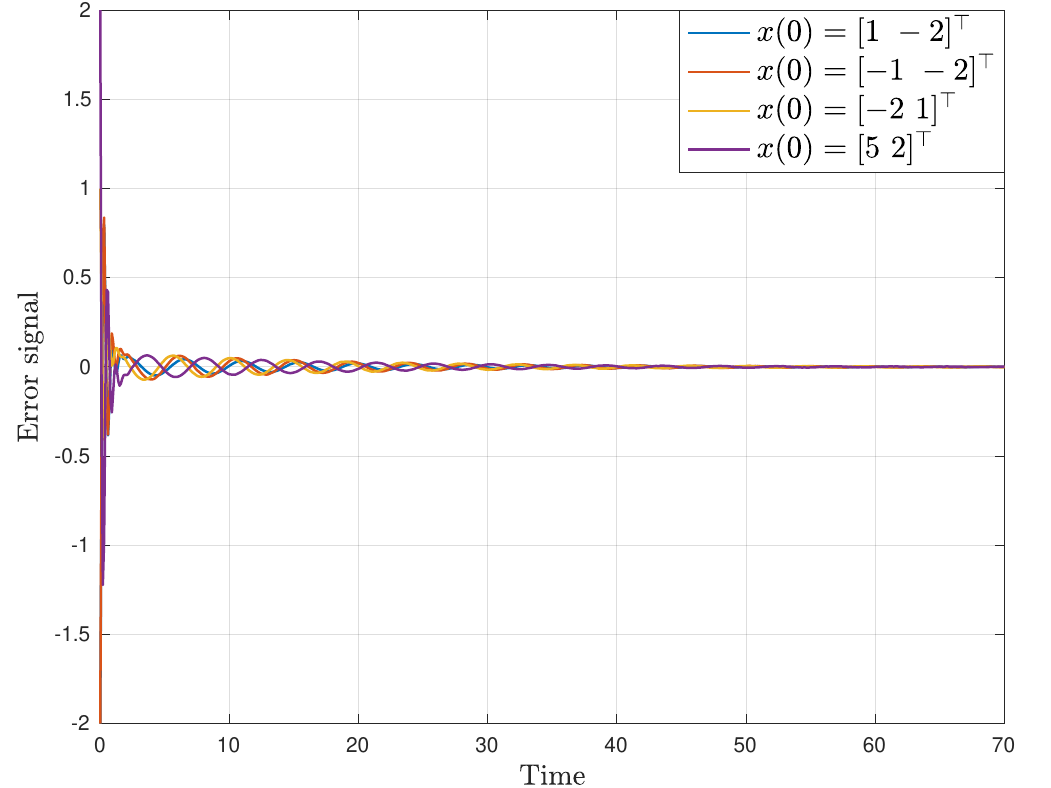}
\caption{{Regulation error of the simple pendulum with various initial conditions.}}
\label{fig:pendulum}
\end{figure}

\section{Data-driven stabilization of non-zero equilibrium}
\label{sec:stabilization}

In industrial applications, common problems such as set-point tracking require stabilizing a non-zero equilibrium. In model-based control, this is often achieved by first performing coordinate transformation on the system model such that the non-zero equilibrium is converted to the origin, and then designing stabilizers for the converted system. However, in the data-driven setting, executing coordinate transformation is challenging, as the equilibrium input is difficult to obtain due to the unknown dynamics. In this section, we use the proposed data-driven output regulation approach for data-driven stabilization of a non-zero equilibrium without knowledge of the equilibrium input. In particular, the designed data-driven controller stabilizes the non-zero equilibrium asymptotically without estimating the unknown equilibrium input despite the presence of uncertain disturbances.

We consider a nonlinear plant with a single input
\begin{align}\label{system:SISO}
    \dot x = A_{\mathrm{p}}Z_{\mathrm{p}}(x) + Bu +E_{\mathrm{p}}w_{\mathrm{p}}
\end{align}
where state $x\in \mathbb{R}^n$, control input $u \in \mathbb{R}$, and $w_{\mathrm{p}} \in \mathbb{R}^{p}$ is the disturbance generated by $\dot w_{\mathrm{p}} = S_{\mathrm{p}} w_{\mathrm{p}}$, satisfying Assumption \ref{assump:S}. 
The control objective is to stabilize the non-zero equilibrium $x_{\mathrm{e}}:= \begin{bmatrix} x_{\mathrm{e}1} & x_{\mathrm{e}2}& \cdots& x_{\mathrm{e}n}\end{bmatrix}^\top$ for system \eqref{system:SISO}. It is emphasized that the corresponding equilibrium input $u_{\mathrm{e}}$ such that 
\begin{align}\label{condi:assignable_equilibrium}
    0 = A_{\mathrm{p}} Z_{\mathrm{p}} (x_{\mathrm{e}})+Bu_{\mathrm{e}}+E_{\mathrm{p}}w_{\mathrm{p}}
\end{align}
is not known. To achieve the control objective, we apply the proposed data-driven output regulation approach, where the output is a virtual error $e_{\mathrm{v}} \in \mathbb{R}$ defined as
\begin{align}
e_{\mathrm{v}} = \|x-x_{\mathrm e}\|^2= \sum_{i=1}^n (x_i - x_{\mathrm{e}i})^2, \label{virtual_error:form squared}
\end{align}
which can be equivalently written as
\begin{align}
e_{\mathrm{v}} = C_{\mathrm{v}} Z_{\mathrm{v}}(x) + w_{\mathrm{v}} \label{virtual_error:form standard}
\end{align}
where $Z_{\mathrm{v}}(x) \in \mathbb{R}^{2n}$ contains all monomials appearing in \eqref{virtual_error:form standard} 
and 
$w_{\mathrm{v}}=||x_{\mathrm{e}}||^2$. 

Let $\widetilde Z(x)$ be a vector that contains all terms in $Z_{\mathrm{p}}(x)$ and $Z_{\mathrm{v}}(x)$, $\widetilde w:=[w_{\mathrm{p}}^\top~~w_{\mathrm{v}}]^\top$, $\widetilde E:=[E_{\mathrm{p}}~~0_{n \times 1} ]$, and $\widetilde F:=[0_{1\times p}~~1]$. The system governed by the dynamics \eqref{system:SISO} and virtual output \eqref{virtual_error:form standard} can be written into the form of \eqref{system:NL} as 
    \begin{align}\label{system:non_zero stab}
    \dot x & = \widetilde A\widetilde Z(x) + Bu + \widetilde E\widetilde w \notag \\ 
     e_{\mathrm{v}} &= \widetilde C \widetilde Z(x) + \widetilde F \widetilde w, 
\end{align}
where $\widetilde C$ is chosen such that $\widetilde C\widetilde Z(x)=C_{\mathrm{v}}Z_{\mathrm{v}}(x)$. Moreover, $\widetilde w$ is governed by dynamics 
\begin{align}\label{exosystem:extended}
    \dot{\widetilde w}= \begin{bmatrix}
        \dot w_{\mathrm{p}} \\ \dot w_{\mathrm{v}}
    \end{bmatrix} = \underbrace{\begin{bmatrix}
    S_{\mathrm{p}} & 0_{p \times 1}\\
    0_{1 \times p} & 0
\end{bmatrix}}_{{:=\widetilde S}} \begin{bmatrix}
        w_{\mathrm{p}} \\ w_{\mathrm{v}}
    \end{bmatrix}.
\end{align}

We have formulated the stabilization problem of \eqref{system:SISO} at $(x_{\mathrm{e}},u_{\mathrm{e}})$ into the same form as the output regulation problem with a known $\widetilde C$. In the following result, we apply the proposed data-driven regulator design approach using the data set $\widetilde{\mathcal{DS}}:=\{\dot x(t_k), x(t_k), u(t_k), k=0,1,\dots,T-1\}$ collected from \eqref{system:non_zero stab}.

\begin{proposition}\label{proposition:nl_knownC}
     Consider system \eqref{system:SISO} and a non-zero equilibrium $x_\mathrm{e}$ {satisfying \eqref{condi:assignable_equilibrium}}, and obtain the extended system \eqref{system:non_zero stab}, the exosystem \eqref{exosystem:extended}, and a sampled data set $\widetilde{\mathcal{DS}}$. If there exist matrices $Y\in\mathbb R^{T\times n_Z}$, $G_2\in\mathbb R^{T}$, and positive definite matrix $P=\mathrm{diag}(P_1,P_2)\in\mathbb R^{n_Z\times n_Z}$ with $P_1\in\mathbb R^{n\times n}$ and $P_2\in\mathbb R^{(n_Z-n)\times (n_Z-n)}$, satisfying
    \begin{subequations}\label{condition:known_C}
        \begin{align}
        \begin{bmatrix}
            P \\ 0_{(p+1)\times n_Z} \end{bmatrix} &=\begin{bmatrix}
        Z_0 \\ M_0
    \end{bmatrix}Y \label{condition:Zx_knownC_Y1Y2} \\
       \begin{bmatrix}
            0_{n_Z\times 1} \\ 1 \\ 0_{(p+1)\times 1}
        \end{bmatrix} &= \begin{bmatrix}
        Z_0 \\ U_0 \\ M_0
    \end{bmatrix}G_2  \label{condition:Zx_knownC_G3} \\
   \mathcal{I}^{\top}
        X_1Y  +(*)^{\top}& \preceq 0  \label{condition:Zx_knownC_AP} \\
    \begin{bmatrix}
        (X_1G_2)^{\top} & 0_{1\times (n_Z-n)}
    \end{bmatrix}&=\widetilde C P \label{condition:Zx_knownC_BC}
        \end{align}
    \end{subequations}
    then the stabilizer 
\begin{align}\label{controller:non_zero_stab}
        \dot \eta &= \widetilde S\eta -\alpha\Xi e_{\mathrm{v}} \notag\\
        u &= K\widetilde Z(x)+\Xi^{\top}\eta -\widehat K e_{\mathrm{v}}
\end{align}
with $K=U_0Y P^{-1}$, any $\alpha>0$ and $\widehat K > 0$ renders the equilibrium $x_{\mathrm{e}}$ globally asymptotically stable. 
\end{proposition}

\begin{proof}
     Substituting $E_0Y$ from Theorem \ref{theorem:nl} with $\widetilde CP$ and $e$ with $e_{\mathrm{v}}$, it follows that the output regulation problem of \eqref{system:non_zero stab} is solved. Consequently, it holds that $\lim_{t\rightarrow \infty} e_{\mathrm v}(t)=0$. By the definition of $e_{\mathrm{v}}$, $x(t)$ asymptotically converges to $x_{\mathrm{e}}$.
\end{proof}

    It is noted that nonlinear data-driven stabilization of non-zero equilibrium has been studied in recent work, such as \cite{Hu2024Contraction,verhoek2023DDDLPV,bisoffi2024setpoint}, and \cite{YLiu2024CDC}. When the equilibrium input $u_{\mathrm{e}}$ is not readily available, an integral action is incorporated into the controller in \cite{verhoek2023DDDLPV} and \cite{Hu2024Contraction}, and $u_{\mathrm{e}}$ is estimated via optimization in \cite{bisoffi2024setpoint}. Proposition \ref{proposition:nl_knownC} overcomes the challenge of unknown $u_{\mathrm{e}}$ 
    using an internal model that also contains integral actions. Nevertheless, the result shown in this section improves the existing results in the sense that it does not require the knowledge or estimation of the equilibrium input, achieves asymptotic stability, and rejects time-varying disturbances.

\begin{remark} [Extension to multiple inputs] 
Proposition \ref{proposition:nl_knownC} can be extended to the case where system \eqref{system:SISO} has multiple inputs, i.e., $u \in \mathbb{R}^m$, if a virtual error $e_{\mathrm{v}} \in \mathbb{R}^m$ can be designed in the form of $e_{\mathrm{v}}=C_{\mathrm{v}}Z(x)+F_{\mathrm{v}}w_{\mathrm{v}}$ such that $\lim_{t\rightarrow\infty}x(t)=x_{\mathrm{e}}$ if and only if $\lim_{t\rightarrow\infty}e_{\mathrm{v}}(t)=0$. In the case where the nonlinear system is fully actuated, i.e., $n=m$, stabilization at a non-zero equilibrium can be achieved straightforwardly by letting $e_{\mathrm{v}} = x-x_{\mathrm{e}}$.
\end{remark}

\textbf{Example 2.} Consider the nonlinear system
\begin{align*}
    &\dot x_1 =2x_2-x_1 \notag \\  
    & \dot x_2 = -x_1+x_2-x_1^2x_2+u+d
\end{align*} 
and the non-zero equilibrium $x_{\mathrm{e}} = [-2~~-1]^\top$, where the disturbance is $d(t)=\cos(2t)$. 

The virtual output $e_{\mathrm{v}}$ is designed as $e_{\mathrm{v}} = (x_1+2)^2 + (x_2+1)^2$, which can be written into the form of \eqref{virtual_error:form standard} with $C_vZ_v(x)=4x_1+2x_2+x_1^2+x_2^2$ and $w_{\mathrm{v}}=5$. The exogenous system generating both $d$ and $w_{\mathrm{v}}$ is \eqref{sys: e.g. exosys} with $d= [0~~1~~0] w$ and $w_{\mathrm{v}}=[0~~0~~1]w$.


To obtain extended system \eqref{system:non_zero stab}, we make $\widetilde Z(x)$ contain all monomials having degree from $1$ to $3$, i.e.,
\begin{align*}
   \!\widetilde Z(x) \!= \!\begin{bmatrix} x_1 & x_2 & x_1x_2 & x_1^2 & x_2^2 & x_1^2x_2 & x_1x_2^2 & x_1^3 & x_2^3 \end{bmatrix}\!^{\top}\! .
\end{align*}
Then, we have $\widetilde C= [4~~2~~1~~1~~0~~0~~0~~0~~0]$ and $\widetilde F =[0~~0~~1]$.

The experimental data is collected by setting $x(t_0)=[-0.1~~0.1]^{\top}$ and $u=\sin(t)$. The data set is obtained with a sampling period $0.5$s and $T=30$. Solving condition \eqref{condition:known_C} with the MOSEK \cite{Mosek2023} solver gives
$K=\begin{bmatrix}
    -36.4774 &-9.7023   & 0 & 0 & 0 & 1 & 0 & 0 & 0
\end{bmatrix}$. Setting $\alpha=30$, $\Xi=[1~~1~~1]^{\top}$, and $\widehat K=20$, the stabilizer \eqref{controller:non_zero_stab} is obtained. Simulation results of state trajectories $x_1(t)$ and $x_2(t)$ with different initial conditions are illustrated in Figure \ref{fig:nonzero_stab}.

\begin{figure}[h]
\centering 
\includegraphics[trim=1.8cm 7cm 2cm 7cm, clip, scale=0.45]{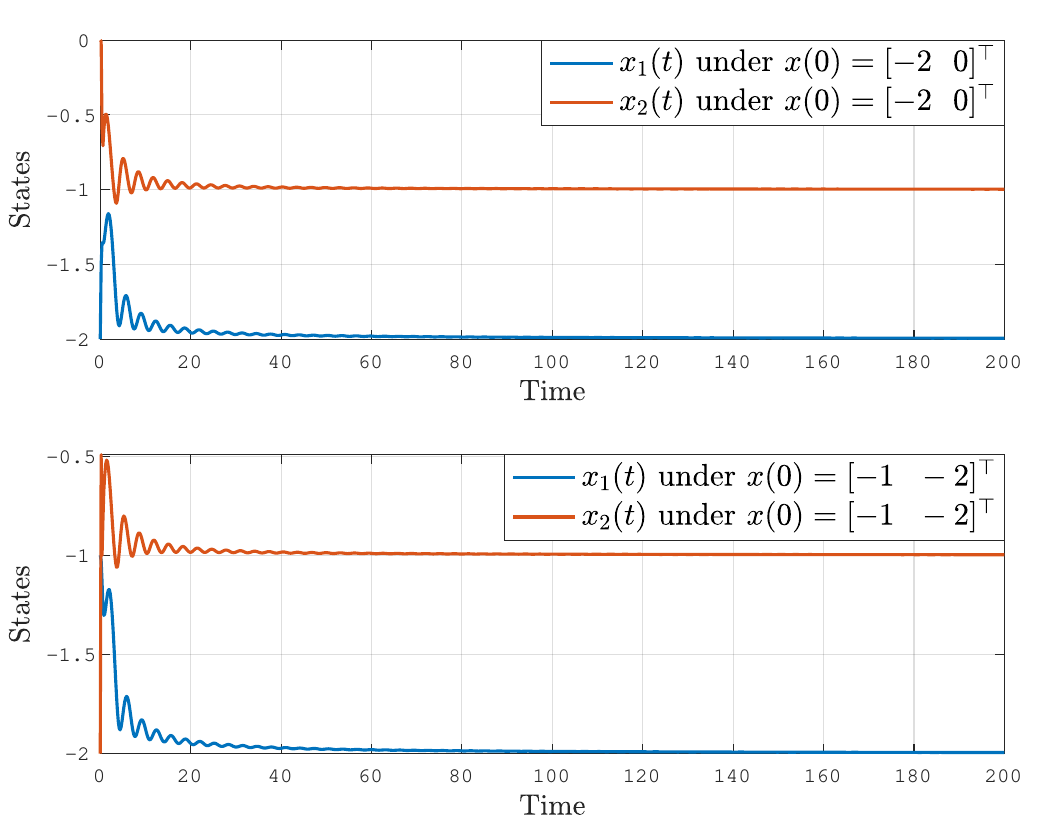}
\caption{State trajectories converging to the non-zero equilibrium $x_{\mathrm{e}}=[-2~~-1]^\top$ with different initial conditions.}
\label{fig:nonzero_stab}
\end{figure}

\vspace{-2ex}
\section{Conclusion and Future Work}\label{sec_conclusion}

This work addressed the data-driven output regulation problem of a class of nonlinear systems using incremental passivity. We first designed a data-driven state feedback controller to render the closed-loop system incrementally passive, then interconnect the closed-loop system with an incrementally passive internal model to complete the regulator design. In particular, a set of data-based LMIs characterize the data-driven feedback passivation controller. The result of this work provides an efficient data-driven regulator design where the passivation controller and the internal model are designed independently and the regulation error asymptotically converges to zero. We also applied the proposed approach to data-driven nonlinear stabilization of non-zero equilibrium without knowledge or estimation of the equilibrium input.  

\emph{Future work.} One of the major challenges to be overcome is handling a more general class of noise. As pointed out in Section \ref{sec:IP}, when the noise cannot be generated by the exosystem in \eqref{exosystem}, the closed-loop system represented using noisy data will contain uncertainties that need to be handled in the passivation process. {We will explore techniques such as robust control and passivity indices to make the closed-loop system incrementally passive despite the system uncertainties brought by noise or estimation errors.
Additionally, as incremental passivity is an input-output property, we will explore extending the proposed approach to data-driven output-feedback passivation and output-feedback regulator design using only input-output data.} 
\vspace{-1ex}

\section*{Acknowledgment}

The authors would like to express their gratitude to Prof. C. De Persis for the helpful discussions related to the results of this paper.    

\bibliographystyle{elsarticle-num}
\bibliography{ref_DDOPRJounal}

\end{document}